\documentclass[12pt]{article}

\usepackage{amsthm}
\usepackage{amsmath}
\usepackage{graphicx}
\usepackage{amssymb}

\newtheorem{thm}{Theorem}[section]
\newtheorem{lma}[thm]{Lemma}
\newtheorem{cor}[thm]{Corollory}

\newtheorem{prop}[thm]{Proposition}

\newtheorem*{adef}{Definition}

\DeclareMathOperator{\aut}{aut}
\DeclareMathOperator{\ColClique}{\#ColClique}
\DeclareMathOperator{\ColSubInd}{\#ColSubInd}
\DeclareMathOperator{\SubInd}{\#SubInd}
\DeclareMathOperator{\StrEmb}{\#StrEmb}
\DeclareMathOperator{\ColStrEmb}{\#ColStrEmb}
\DeclareMathOperator{\clique}{\#Clique}
\DeclareMathOperator{\constr}{constr}

\begin{document}

% Title portion
\title{Some hard families of parameterised counting problems \thanks{Research supported by EPSRC grant ``Computational Counting''}}
\date{September 2014}
\author{Mark Jerrum and Kitty Meeks \\
\small{School of Mathematical Sciences, Queen Mary University of London} \\ 
\texttt{\small{\{m.jerrum,k.meeks\}@qmul.ac.uk}}}
\maketitle

\begin{abstract}
We consider parameterised subgraph-counting problems of the following form: given a graph $G$, how many $k$-tuples of its vertices induce a subgraph with a given property?  A number of such problems are known to be \#W[1]-complete; here we substantially generalise some of these existing results by proving hardness for two large families of such problems.  We demonstrate that it is \#W[1]-hard to count the number of $k$-vertex subgraphs having any property where the number of distinct edge-densities of labelled subgraphs that satisfy the property is $o(k^2)$.  In the special case that the property in question depends only on the number of edges in the subgraph, we give a strengthening of this result which leads to our second family of hard problems.
\end{abstract}

\newcommand{\leqfptT}{ $\leq^{\textup{fpt}}_{\textup{T}}$ }
\newcommand{\leqfptP}{ $\leq^{\textup{fpt}}_{\textup{pars}}$ }
\newcommand{\paramcount}[1]{\textup{\textbf{p-\#}}\textsc{#1}}
\newcommand{\paramdec}[1]{\textup{\textbf{p-}}\textsc{#1}}
\newcommand{\genprob}{Induced Subgraph With Property}
\newcommand{\genprobcol}{Multicolour Induced Subgraph with Property}
\newcommand{\construct}{\constr(G,f_G,H,U)}
\newcommand{\constructind}{\overline{\constr}(G,f_G,H,W)}

\section{Introduction}

Parameterised counting problems were introduced by Flum and Grohe in \cite{flum04}.  Much previous research has focussed on problems of the following form:

\begin{quote}
\textit{Input:} An $n$-vertex graph $G = (V,E)$, and an integer $k$. \\
\textit{Parameter:} $k$. \\
\textit{Question:} How many (labelled) $k$-vertex subsets of $V$ induce graphs with a given property?
\end{quote}

All the existing literature concerning the complexity of solving non-trivial problems of this kind exactly consists of \#W[1]-completeness results, implying that the problems considered are unlikely to be solvable in time $f(k)n^{O(1)}$ for any function $f$; \footnote{See Section \ref{param} for definitions of concepts from parameterised complexity; we will also see in Section \ref{model} that all problems of this specific form belong to the first level of the W-hierarchy.} non-trivial in this sense means that there is no constant $c$ so that, for any $k \in \mathbb{N}$, we can determine whether a given graph has the desired property by examining only edges incident with some fixed set of $c$ vertices (a dichotomy result for a special class of these problems was very recently proved by Curticapean and Marx \cite{radu14}, in which parameterised tractability coincides exactly with this definition of triviality).  A number of these results concern the complexity of \emph{induced} subgraph counting problems: Chen and Flum \cite{chen07} demonstrated that problems of counting $k$-vertex induced paths and of counting $k$-vertex induced cycles are both \#W[1]-complete, and more generally Chen, Thurley and Weyer \cite{chen08} showed that it is \#W[1]-complete to count the number of induced subgraphs isomorphic to a given graph from the class $\mathcal{C}$ (\paramcount{Induced Subgraph Isomorphism}$(\mathcal{C})$) whenever $\mathcal{C}$ contains arbitrarily large graphs.  Other results concern the complexity of ``non-induced" subgraph counting problems, including the problems of counting the number of paths (\paramcount{Path}) and cycles (\paramcount{Cycle}) \cite{flum04}, matchings (\paramcount{Matching} \cite{radu13}), and connected subgraphs (\paramcount{Connected Induced Subgraph} \cite{connected}); the well-studied problem of counting the number of $k$-vertex cliques (\paramcount{Clique} \cite{flum04}) can be considered as either an induced or non-induced subgraph problem.  However, even considering these examples, the number of problems known to be complete for the parameterised complexity class \#W[1] as a whole remains relatively small.

In this paper, we add to this collection of hard parameterised counting problems by giving two conditions, either of which is sufficient to guarantee that a subgraph counting problem of this kind is \#W[1]-complete.  The two resulting families of hard parameterised subgraph counting problems contain some of the special cases already known to be hard (including \paramcount{Induced Subgraph Isomorphism}$(\mathcal{C})$) but are defined in a very general way and so include many problems whose complexity status was not previously known.

The precise formulation of our results makes use of the general model for parameterised subgraph counting problems introduced in \cite{connected} and described in Section \ref{model} below, but informally we show that counting labelled induced subgraphs with property $\Phi$ is \#W[1]-complete in each of the following situations:
\begin{enumerate}
\item $D_k = \{|E(H)|: |V(H)| = k$ and $\Phi$ is true for $H \}$ satisfies $|D_k| = o(k^2)$, that is, the property $\phi_k$ holds only for a decreasing proportion of the possible edge densities (Theorem \ref{bounded-hard}).
\item $\Phi$ is defined by a collection of $o(k^2)$ sub-intervals of $\{0,\ldots,\binom{k}{2}\}$, such that $\Phi$ is true for $H$ if and only if the number of edges in $H$ lies in one of these intervals (Theorem \ref{interval-hard}).
\end{enumerate}
The first class of problems includes those of counting $k$-vertex induced subgraphs which are planar, or have treewidth at most $t$ (for any fixed $t$), as any subgraph with either of these properties has $o(k)$ edges; it also includes the problem of counting the number of regular $k$-vertex subgraphs (that is, graphs that are $d$-regular for any $d \in \{0,\ldots,k-1\}$), as there are only $k$ possible edge-densities for a regular graph on $k$ vertices.  Problems that belong to the second class but not the first include, for example, counting all $k$-vertex subgraphs with edge-density at least $\alpha$, where alpha is some constant in $[0,1]$ that does not depend on $k$.

The proofs of our results will use ideas from Ramsey theory.  This field of extremal graph theory has previously been exploited to prove hardness results for parameterised counting problems, for example in \cite{chen08}.  In this paper, we need a different kind of Ramsey theoretic result that guarantees more than just the existence of a single clique or independent set, and to this end we derive a lower bound on the total number of $k$-vertex cliques and independent sets that must be present in any $n$-vertex graph (if $n$ is sufficiently large compared with $k$).

The rest of the paper is organised as follows.  In the remainder of this section, we introduce our key notation and definitions, mention the main ideas we will use from the theory of parameterised complexity, prove our Ramsey theoretic result, and finally give a formal definition of the model for subgraph counting problems.  In Section \ref{construction}, we define a pair of closely related constructions which form the basis of our later reductions, and demonstrate the important properties of these constructions.  Section \ref{hard} then contains the proofs of our \#W[1]-hardness results.

\subsection{Notation and definitions}
\label{notation}

Given a graph $G = (V,E)$, and a subset $U \subseteq V$, we write $G[U]$ for the subgraph of $G$ induced by the vertices of $U$.  For any $k \in \mathbb{N}$, we write $[k]$ as shorthand for $\{1,\ldots,k\}$, and $V^{(k)}$ for the set of all subsets of $V$ of size exactly $k$.  A \emph{permutation} on $[k]$ is a bijection $[k] \rightarrow [k]$.  We denote by $\overline{G}$ the \emph{complement} of $G$, that is, $\overline{G} = (V,E')$ where $E' = V^{(2)} \setminus E$. 

Two graphs $G$ and $H$ are \emph{isomorphic}, written $G \cong H$, if there exists a bijection $\theta: V(G) \rightarrow V(H)$ so that, for all $u,v \in V(G)$, we have $\theta(u)\theta(v) \in E(H)$ if and only if $uv \in E(G)$; $\theta$ is said to be an \emph{isomorphism} from $G$ to $H$.  An \emph{automorphism} on $G$ is an isomorphism from $G$ to itself.  We write $\aut(G)$ for the number of automorphisms of $G$.

If $G$ is coloured by some colouring $f: V \rightarrow [k]$, we say that a subset $U \subseteq V$ is \emph{colourful} (under $f$) if, for every $i \in [k]$, there exists exactly one vertex $u \in U$ such that $f(u) = i$; note that this can only be achieved if $U \in V^{(k)}$.

We will be considering labelled graphs, where a labelled graph is a pair $(H, \pi)$ such that $H$ is a graph and $\pi : [|V(H)|] \rightarrow V(H)$ is a bijection.  We write $\mathcal{L}(k)$ for the set of all labelled graphs on the vertex set $[k]$.  Given a graph $G = (V,E)$ and a $k$-tuple of vertices $(v_1,\ldots,v_k)$, $G[v_1,\ldots,v_k]$ denotes the labelled graph $(H,\pi)$ where $H = G[\{v_1,\ldots,v_k\}]$ and $\pi(i) = v_i$ for each $i \in [k]$.  If $\mathcal{H}$ is a set of labelled graphs, we set $\mathcal{H}^H = \{(H',\pi') \in \mathcal{H}: H' \cong H\}$.  

Given two graphs $G$ and $H$, a \emph{strong embedding} of $H$ in $G$ is an injective mapping $\theta: V(H) \rightarrow V(G)$ such that, for any $u,v \in V(H)$, $\theta(u)\theta(v) \in E(G)$ if and only if $uv \in E(H)$.  We denote by $\StrEmb(H,G)$ the number of strong embeddings of $H$ in $G$.  If $\mathcal{H}$ is a class of labelled graphs on $k$ vertices, we set 
\begin{align*}
\StrEmb(\mathcal{H},G) =  \qquad \qquad  & \\
|\{\theta: [k] \rightarrow V(G) \quad : \quad & \theta \text{ is injective and } \exists (H,\pi) \in \mathcal{H}  \text{ such that } \\ &  \theta(i)\theta(j) \in E(G) \iff \pi(i)\pi(j) \in E(H)\}|.
\end{align*}
If $G$ is also equipped with a $k$-colouring $f$, where $|V(H)| = k$, we write $\ColStrEmb(H,G,f)$ for the number of strong embeddings of $H$ in $G$ such that the image of $V(H)$ is colourful under $f$.  Similarly, we set 
\begin{align*}
\ColStrEmb(\mathcal{H},G,f) = \qquad \qquad & \\
 |\{\theta:[k] \rightarrow V(G) \quad : \quad & \theta \text{ is injective, } \exists (H,\pi) \in \mathcal{H} \text{ such that } \\ 
 						& \theta(i)\theta(j) \in E(G) \iff \pi(i)\pi(j) \in E(H), \\
 						& \text{and $\theta([k])$ is colourful under } f\}|.
\end{align*}
We can alternatively consider unlabelled embeddings of $H$ in $G$.  In this context we write $\SubInd(H,G)$ for the number of subsets $U \in V(G)^{(|H|)}$ such that $G[U] \cong H$.  Note that $\SubInd(H,G) = \StrEmb(H,G) / \aut(H)$.  If $\mathcal{H}$ is a class of labelled graphs, we set 
\begin{align*}
\SubInd(\mathcal{H},G) = |\{U \subseteq V(G) : & \quad \exists (H,\pi) \in \mathcal{H} \text{ such that } G[U] \cong H\}|. 
\end{align*} 
Once again, we can also consider the case in which $G$ is equipped with a $k$-colouring $f$.  In this case $\ColSubInd(H,G,f)$ is the number of colourful subsets $U$ such that $G[U] \cong H$, and 
\begin{align*}
\ColSubInd(\mathcal{H},G,f) = |\{U \subseteq V(G) \quad : \quad & \exists (H,\pi) \in \mathcal{H} \text{ such that } G[U] \cong H,  \\
                         & \text{ and $U$ is colourful under $f$}\}|.
\end{align*} 
Finally, we write $\ColClique_k(G,f)$ as shorthand for $\ColSubInd(K_k,G,f)$, where $K_k$ denotes a clique on $k$ vertices.

\subsection{Parameterised Counting Complexity}
\label{param}

In this section, we introduce key notions from parameterised counting complexity, which we will use in the rest of the paper.  A parameterised counting problem is a pair $(\Pi,\kappa)$ where, for some finite alphabet $\Sigma$, $\Pi: \Sigma^* \rightarrow \mathbb{N}_0$ is a function  and $\kappa: \Sigma^* \rightarrow \mathbb{N}$ is a parameterisation (a polynomial-time computable mapping).  An algorithm $A$ for a parameterised counting problem $(\Pi,\kappa)$ is said to be an \emph{fpt-algorithm} if there exists a computable function $f$ and a constant $c$ such that the running time of $A$ on input $I$ is bounded by $f(\kappa(I))|I|^c$.  Problems admitting an fpt-algorithm are said to belong to the class FPT.

To understand the complexity of parameterised counting problems, Flum and Grohe \cite{flum04} introduce two kinds of reductions between such problems.

\begin{adef}
Let $(\Pi,\kappa)$ and $(\Pi',\kappa')$ be parameterised counting problems.
\begin{enumerate}
\item An fpt parsimonious reduction from $(\Pi,\kappa)$ to $(\Pi',\kappa')$ is an algorithm that computes, for every instance $I$ of $\Pi$, an instance $I'$ of $\Pi'$ in time $f(\kappa(I))\cdot |I|^c$ such that $\kappa'(I') \leq g(\kappa(I))$ and 
$$\Pi(I) = \Pi'(I')$$ 
(for computable functions $f,g: \mathbb{N} \rightarrow \mathbb{N}$ and a constant $c \in \mathbb{N}$).  In this case we write $(\Pi,\kappa)$ \leqfptP $(\Pi',\kappa')$.

\item An fpt Turing reduction from $(\Pi,\kappa)$ to $(\Pi',\kappa')$ is an algorithm $A$ with an oracle to $\Pi'$ such that
\begin{enumerate}
\item $A$ computes $\Pi$,
\item $A$ is an fpt-algorithm with respect to $\kappa$, and
\item there is a computable function $g:\mathbb{N} \rightarrow \mathbb{N}$ such that for all oracle queries ``$\Pi'(I') = \; ?$'' posed by $A$ on input $I$ we have $\kappa'(I') \leq g(\kappa(I))$.
\end{enumerate}
In this case we write $(\Pi,\kappa)$ \leqfptT $(\Pi',\kappa')$.
\end{enumerate}

\end{adef}

Using these notions, Flum and Grohe introduce a hierarchy of parameterised counting complexity classes, \#W[$t$], for $t \geq 1$; this is the analogue of the W-hierarchy for parameterised decision problems.  In order to define this hierarchy, we need some more notions related to satisfiability problems.  

The definition of levels of the hierarchy uses the following problem, where $\psi$ is a first-order formula with a free relation variable of arity $s$.
\\

\hangindent=1cm
\paramcount{WD}$_{\psi}$ \\
\textit{Input:} A structure $\mathcal{A}$ and $k \in \mathbb{N}$. \\
\textit{Parameter:} $k$. \\
\textit{Question:} How many relations $S \subseteq A^s$ of cardinality $|S|=k$ are such that $\mathcal{A} \models \psi(S)$ (where $A$ is the universe of $\mathcal{A}$)? \\

If $\Psi$ is a class of first-order formulas, then \paramcount{WD}-$\Psi$ is the class of all problems \paramcount{WD}$_{\psi}$ where $\psi \in \Psi$.  The classes of first-order formulas $\Sigma_t$ and $\Pi_t$, for $t \geq 0$, are defined inductively.  Both $\Sigma_0$ and $\Pi_0$ denote the class of quantifier-free formulas, while, for $t \geq 1$, $\Sigma_t$ is the class of formulas
$$\exists x_1 \ldots \exists x_i \psi,$$
where $\psi \in \Pi_{t-1}$, and $\Pi_t$ is the class of formulas
$$\forall x_1 \ldots \forall x_i \psi,$$
where $\psi \in \Sigma_{t-1}$.  We are now ready to define the classes \#W[$t$], for $t \geq 1$.

\begin{adef}[Flum and Grohe \cite{flum04,flumgrohe}]
For $t \geq 1$, \#W[$t$] is the class of all parameterised counting problems that are fpt parsimonious reducible to a problem in \paramcount{WD}-$\Pi_t$.
\end{adef} 

Just as it is considered to be very unlikely that W[1] = FPT, it is very unlikely that there exists an algorithm running in time $f(k)n^{O(1)}$ for any problem that is hard for the class \#W[1] under either fpt parsimonious reductions or fpt Turing reductions.  One useful \#W[1]-complete problem which we will use in our reductions is the following:
\\

\hangindent=1cm
\paramcount{Multicolour Clique} \\
\textit{Input:} A graph $G = (V,E)$, and a $k$-colouring $f$ of $G$. \\
\textit{Parameter:} $k$. \\
\textit{Question:} What is $\ColClique_k(G,f)$, that is, the number of $k$-vertex cliques in $G$ that are colourful with respect to $f$? \\

This problem can easily be shown to be \#W[1]-hard (along the same lines as the proof of the W[1]-hardness of \paramdec{Multicolour Clique} in \cite{fellows09}) by means of a reduction from \paramcount{Clique}, shown to be \#W[1]-hard in \cite{flum04}.

\subsection{Ramsey theory}
\label{extremal}

To show that our constructions have the desired properties, we will exploit some Ramsey theoretic results.  First of all, we will use the following bound on Ramsey numbers which follows immediately from a result of Erd\H{o}s and Szekeres \cite{erdos-szekeres}:
\begin{thm}
Let $k \in \mathbb{N}$.  Then there exists $R(k) < 2^{2k}$ such that any graph on $n \geq R(k)$ vertices contains either a clique or independent set on $k$ vertices.
\label{ramsey}
\end{thm}
We will also need the following easy corollary of this result.
\begin{cor}
Let $G = (V,E)$ be an $n$-vertex graph, where $n \geq 2^{2k}$.  Then the number of $k$-vertex subsets $U \subset V$ such that $U$ induces either a clique or independent set in $G$ is at least
$$\frac{(2^{2k} - k)!}{(2^{2k})!}\frac{n!}{(n-k)!}.$$
\label{ramsey-cor}
\end{cor}
\begin{proof}
We shall say that the subset $X \in V^{(k)}$ is \emph{interesting} if $X$ induces either a clique or an independent set in $G$.  By Ramsey's Theorem, we know that every subset $U \subset V$ with $|U| = 2^{2k}$ must contain at least one interesting subset.

The number of subsets of $V$ of size exactly $2^{2k}$ is $\binom{n}{2^{2k}}$.  Moreover, the number of such sets to which any given $k$-vertex subset can belong is $\binom{n-k}{2^{2k} - k}$.  Thus, in order for every $U \in V^{(2^k)}$ to contain at least one interesting subset, the number of interesting subsets must be at least
\begin{align*}
\frac{\binom{n}{2^{2k}}}{\binom{n-k}{2^{2k}-k}} = \frac{\frac{n!}{(2^{2k})!(n-2^{2k})!}}{\frac{(n-k)!}{(2^{2k} - k)!(n-2^{2k})!}}
 = \frac{(2^{2k} - k)!}{(2^{2k})!}\frac{n!}{(n-k)!},
\end{align*}
as required.
\end{proof}

\subsection{The Model}
\label{model}

The classes of counting problems we consider fall within the scope of the general model introduced in \cite{connected}; this model describes parameterised counting problems in which the goal is to count labelled subgraphs with particular properties.  We repeat the definition here for completeness, before extending it to colourful subgraph counting problems (which we will need for intermediate stages in our reductions).  We will finish with some examples of how problems that have previously been studied in the literature can be expressed in this framework.

Let $\Phi$ be a family $(\phi_1,\phi_2,\ldots)$ of functions $\phi_k: \mathcal{L}(k) \rightarrow \{0,1\}$, such that the function mapping $k \mapsto \phi_k$ is computable.  For any $k$, we write $\mathcal{H}_{\phi_k}$ for the set $\{(H,\pi) \in \mathcal{L}(k): \phi_k(H,\pi) = 1\}$, and set $\mathcal{H}_{\Phi} = \bigcup_{k \in \mathbb{N}} \mathcal{H}_{\phi_k}$.

The general problem is then defined as follows.
\\

\hangindent=1cm
\paramcount{\genprob}($\Phi$) \\
\textit{Input:} A graph $G = (V,E)$ and an integer $k$.\\
\textit{Parameter:} $k$. \\
\textit{Question:} What is $\StrEmb(\mathcal{H}_{\phi_k},G)$, that is, the cardinality of the set $\{(v_1,\ldots,v_k) \in V^k: \phi_k(G[v_1,\ldots,v_k]) = 1 \}$? \\

We can equivalently regard this problem as that of counting induced labelled $k$-vertex subgraphs that belong to $\mathcal{H}_{\Phi}$.

A property $\Phi$ is said to be \emph{symmetric} if the value of $\phi_k(H,\pi)$ depends only on the graph $H$ and not on the labelling of the vertices; this corresponds to ``unlabelled'' graph problems, such as \paramcount{clique}.  A related problem for symmetric properties was also defined in \cite{connected}:
\\

\hangindent=1cm
\paramcount{Induced Unlabelled Subgraph With Property}($\Phi$) \\
\textit{Input:} A graph $G = (V,E)$ and $k \in \mathbb{N}$.\\
\textit{Parameter:} $k$. \\
\textit{Question:} What is $\SubInd(\mathcal{H}_{\phi_k})$, that is, the cardinality of the set $\{\{v_1,\ldots,v_k\} \in V^{(k)}: \phi_k(G[v_1,\ldots,v_k]) = 1 \}$? \\

For any symmetric property $\Phi$, the output of \paramcount{\genprob}$(\Phi)$ is exactly $k!$ times the output of \paramcount{Induced Unlabelled Subgraph With Property}$(\Phi)$, for any graph $G$ and $k \in \mathbb{N}$.

These problems were shown to lie in \#W[1] in \cite{connected}:
\begin{prop}[\cite{connected}]
For any $\Phi$, the problem \paramcount{\genprob}($\Phi$) belongs to \#W[1].  If $\Phi$ is symmetric, then the same is true for \paramcount{Induced Unlabelled Subgraph With Property}$(\Phi)$.
\label{in-W}
\end{prop}

We now observe that the complexities of solving \paramcount{\genprob}$(\Phi)$ and its complement, \paramcount{\genprob}$(\overline{\Phi})$, must be the same.
\begin{prop}
Let $\Phi$ be a family $(\phi_1,\phi_2,\ldots)$ of functions $\phi_k: \mathcal{L}(k) \rightarrow \{0,1\}$, such that the function mapping $k \mapsto \phi_k$ is computable, and let $\overline{\Phi}$ be the family $(\overline{\phi_1},\overline{\phi_2},\ldots)$ of functions $\overline{\phi_k}: \mathcal{L}(k) \rightarrow \{0,1\}$ defined by $\overline{\phi_k}(H,\pi) = 1 - \phi_k(H,\pi)$.  Then \paramcount{\genprob}$(\Phi)$ \leqfptT \paramcount{\genprob}$(\overline{\Phi})$.
\label{complement}
\end{prop}
\begin{proof}
Observe that, for any graph $G$,
\begin{align*}
\StrEmb(\mathcal{H}_{\phi_k},G) & = \sum_{(v_1,\ldots,v_k) \in V(G)^k} \phi_k(G[v_1,\ldots,v_k]) \\
								 & = \binom{n}{k} k! - \sum_{(v_1,\ldots,v_k) \in V(G)^k} \overline{\phi_k}(G[v_1,\ldots,v_k]) \\
								 & = \binom{n}{k} k! - \StrEmb(\mathcal{H}_{\overline{\phi_k}},G).
\end{align*}
Thus it is clear that we can solve \paramcount{\genprob}($\Phi$) in polynomial time using a single oracle call to \paramcount{\genprob}($\Phi$), where the parameter value is the same for both problems; this completes the reduction.
\end{proof}
This result implies that, if \paramcount{\genprob}$(\Phi)$ is in FPT, so is \paramcount{\genprob}$(\overline{\Phi})$, and if \paramcount{\genprob}$(\Phi)$ is W[1]-complete under fpt Turing reductions then so is \paramcount{\genprob}$(\overline{\Phi})$.

We also define a multicolour version of this problem; it is straightforward to verify that this variant also lies in the class \#W[1] for every property $\Phi$.
\\

\hangindent=1cm
\paramcount{\genprobcol}($\Phi$) \\
\textit{Input:} A graph $G = (V,E)$, an integer $k$ and colouring $f: V \rightarrow [k]$.\\
\textit{Parameter:} $k$. \\
\textit{Question:} What is $\ColStrEmb(\mathcal{H}_{\phi_k},G,f)$, that is, the cardinality of the set $\{(v_1,\ldots,v_k) \in V^k: \phi_k(G[v_1,\ldots,v_k]) = 1$ and $\{f(v_1),\ldots,f(v_k)\} = [k] \}$? \\

We make the following simple observation regarding the complexities of \paramcount{\genprobcol}($\Phi$) and \paramcount{\genprob}($\Phi$).

\begin{lma}
For any family $\Phi$, we have \paramcount{\genprobcol}$(\Phi)$ \leqfptT \paramcount{\genprob}$(\Phi)$.
\label{uncol-col}
\end{lma}
\begin{proof}
We give an fpt Turing reduction from \paramcount{\genprobcol}($\Phi$) to \paramcount{\genprob}($\Phi$), using an inclusion-exclusion method (similar to those previously used in, for example, \cite{chen08,dalmau04}).  Let $G$, with colouring $f$, be the $k$-coloured graph in an instance of \paramcount{\genprobcol}$(\Phi)$.  

Suppose we have an oracle to \paramcount{\genprob}($\Phi$), so for any $G' = (V_{G'},E_{G'})$ and $k' \in \mathbb{N}$ we can obtain the cardinality of the set
$$X_{G'} = \{(v_1,\ldots,v_{k'}) \in V_{G'}^{k'}: \phi_{k'}(G'[v_1,\ldots,v_{k'}]) = 1 \}$$
in constant time.  Our goal is to compute the cardinality of the set
\begin{align*}
Y = \{(v_1,\ldots,v_k) \in V^k \quad : \quad & \phi_k(G[v_1,\ldots,v_k]) = 1 \text{ and } \\
& \{f(v_1),\ldots,f(v_k)\} = [k] \}.
\end{align*}
It is clear that, if for each $I \subseteq [k]$ we set 
$$N_I = |\{(v_1,\ldots,v_k) \in V^k: \phi_k(G[v_1,\ldots,v_k]) \text{ and } \{f(v_1),\ldots,f(v_k)\} \subseteq I\}|,$$
then the cardinality of $Y$ can be written as
$$|Y| = \sum_{I \subseteq [k]} (-1)^{k - |I|}N_I.$$
But for any $I \subseteq [k]$, we have
$$N_I = |X_{G[f^{-1}(I)]}|,$$
that is, $N_I$ is equal to the number of tuples of vertices in $G$ which satisfy $\phi_k$ and are such that all the vertices have colours from $I$.  Thus we can compute each of the $2^k$ values of $N_I$ for $I \subseteq [k]$ in time $n^{O(1)}$ using an oracle call, and the parameter for each oracle call is exactly $k$.  This gives an fpt Turing reduction from \paramcount{\genprobcol}($\Phi$) to \paramcount{\genprob}($\Phi$).
\end{proof}

\subsubsection*{Examples}
\label{model-egs}

We now give some examples to illustrate how, up to known constant factors depending only on the parameter, our model generalises previously studied problems.

\paragraph*{A symmetric property: \paramcount{Clique}}

Set 
\begin{equation*}
\phi_k(H,\pi) = \begin{cases}
					1	& \text{if $H \cong K_k$} \\
					0   & \text{otherwise.}
			    \end{cases}
\end{equation*}
Then the output of \paramcount{\genprob}$(\Phi)$ on the input $(G,k)$ is equal to $k!$ times the output of \paramcount{Clique} on the same input; in this situation the outputs of \paramcount{Induced Unlabelled Subgraph With Property}$(\Phi)$ and \paramcount{Clique} will be identical.

\paragraph*{An induced subgraph counting problem: \paramcount{Induced Subgraph Isomorphism}$(\mathcal{C})$}

Let $\mathcal{C}$ be any recursively enumerable class of unlabelled graphs which contains at most one graph on $k$ vertices, for each $k \in \mathbb{N}$.  We set
\begin{equation*}
\phi_k(H,\pi) = \begin{cases}
					1	& \text{if $H \in \mathcal{C}$} \\
					0	& \text{otherwise.}
				\end{cases}
\end{equation*}
Then, for any $\{v_1,\ldots,v_k\} \in V^{(k)}$ such that $G[\{v_1,\ldots,v_k\}] \cong C \in \mathcal{C}$, we will have $\phi_k(G[v_{\sigma(1)},\ldots,v_{\sigma(k)}] = 1$ for every permutation $\sigma: [k] \rightarrow [k]$.  Thus, the output of \paramcount{\genprob}$(\Phi)$ on the input $(G,k)$ is equal to $k!$ times the number of $k$-vertex induced subgraphs in $G$ that belong to $\mathcal{C}$.

\paragraph*{A non-induced subgraph counting problem: \paramcount{Sub}$(\mathcal{H})$}

The problem \paramcount{Sub}$(\mathcal{H})$ is that of counting (not necessarily induced) copies of graphs from a set $\mathcal{H} = \{H_k: k \in I_{\mathcal{H}} \subseteq \mathbb{N}\}$ where, for each $k$, $H_k$ has $k$ vertices.  We begin with a concrete example, \paramcount{Matching}.  Here, $I_{\mathcal{H}}$ is the set of all even natural numbers and, for each $k \in I_{\mathcal{H}}$, $H_k$ is the graph consisting of $k/2$ disjoint edges.  We can then set
\begin{equation*}
\phi_k(H,\pi) = \begin{cases}
					1	& \text{if $k$ is even and, for $1 \leq i \leq k/2$, $\pi(2i-1)\pi(2i) \in E(H)$} \\
					0	& \text{otherwise,}
				\end{cases}
\end{equation*}
and the output of \paramcount{\genprob}$(\Phi)$ on the input $(G,k)$ will be equal to $(k/2)!2^{k/2}$ times the output of \paramcount{Matching} (since $(k/2)!2^{k/2}$ is the number of automorphisms of a $k/2$-edge matching: there are $(k/2)!$ ways to map a set of $k/2$ edges to itself, and each edge can be mapped to any given edge in two different ways).

More generally, to count copies of graphs from $\mathcal{H}$ we fix, for each $k \in I_{\mathcal{H}}$, a labelling $\pi_k: [k] \rightarrow V(H_k)$, and set 
\begin{equation*}
\phi_k(H,\pi) = \begin{cases}
					1	& \text{if $k \in I_{\mathcal{H}}$ and, for every $uv \in E(H_k)$,}\\
						& \qquad \qquad \text{we have } (\pi \circ \pi_k^{-1}(u)) (\pi \circ \pi_k^{-1}(v))  \in E(H)\\
					0	& \text{otherwise.}
				\end{cases}
\end{equation*}
The output of \paramcount{\genprob}$(\Phi)$ on the input $(G,k)$ is then equal to $\aut(H_k)$ times the output of \paramcount{Sub}$(\mathcal{H})$.

It will follow from Theorem \ref{bounded-hard} that \paramcount{Sub}$(\mathcal{H})$ is hard for any class $\mathcal{H} = \{H_k: k \in I_H\}$ such that $e(H_k) = (1 - o(1))\binom{k}{2}$ as $k \rightarrow \infty$.

Observe that problems of this kind are the first examples of problems for which we need our model to count \emph{labelled} subgraphs: if we were only able to count unlabelled subgraphs, we could not give different weight to induced subgraphs containing different numbers of distinct copies of graphs from $\mathcal{H}$, and could only define properties corresponding to induced $k$-vertex subgraphs that contain at least $r$ copies of $H_k$.  For example, we could express the problem of counting the number of induced $k$-vertex subgraphs that contain at least one perfect matching using only unlabelled subgraphs, but to translate \paramcount{Matching} into such a framework we need to make use of the labelling.

\section{The construction}
\label{construction}

In this section we describe a pair of closely related constructions, which will be used for hardness reductions in Section \ref{hard}.  Both constructions take as input two graphs $G$ and $H$, where $G=(V_G,E_G)$ is equipped with a $k$-colouring $f_G: V_G \rightarrow [k]$, and $H$ contains either a clique or independent set on $k$ vertices; the two different constructions correspond to these two possibilities for $H$.  We begin in Section \ref{const-def} by describing the constructions in both cases, and then in Section \ref{const-props} we prove a number of key facts about the two constructions.

\subsection{Definition of the construction}
\label{const-def}

As explained above, we give two slightly different constructions depending on whether $H$ contains a clique or an independent set on $k$ vertices.  We begin with the former case.

\subsubsection*{$H$ contains a clique}

In this case we assume that there exists a set $U \in V(H)^{(k)}$ that induces a clique in $H$.  We now define a new graph, $\construct$, and a colouring $f_{\construct}$ of its vertices.  Suppose that $H = (V_H,E_H)$, renaming vertices if necessary so that $V_H \cap V_G = \emptyset$; set $V_H' = V_H \setminus U$.  We then set
$$V(\construct) = V_G \cup V_H'.$$
Let $f_H: V_H \rightarrow [|V_H|]$ be any colouring of $V_H$ which assigns a distinct colour to each vertex, and which has the property that, for every $u \in U$, $f_H(u) \in [k]$.  Different choices of this colouring $f_H$ may result in different graphs $\construct$, but the properties of $\construct$ that we will exploit hold regardless of the choice of $f_H$, provided that the function satisfies these requirements.  We will set $E(\construct) = E_G \cup E_1 \cup E_2$ where
$$E_1 = \{uv \in E_H: u,v \in V_H' \},$$
and
\begin{align*}
E_2 = \{vw \quad : \quad & v \in V_G, w \in V_H', \exists u \in U \text{ such that } \\
& uw \in E_H \text{ and } f_H(u) = f_G(v)\},
\end{align*}
so $\construct$ contains all internal edges in $G$ and $H \setminus U$, together with edges from each vertex $w$ in $V_G$ to the vertices in $V_H'$ which are adjacent, in $H$, to the vertex of $U$ assigned colour $f_G(w)$ by $f_H$.

Finally, we define the colouring $f_{\construct}: V(\construct) \rightarrow [|V(H)|]$ by setting
\[
f_{\construct}(v) = \begin{cases}
		  f_H(v) & \text{if $v \in V_H'$} \\
		  f_G(v) & \text{if $v \in V_G$.}
	   \end{cases}
\]

\subsubsection*{$H$ contains an independent set}  

In this case we assume that there exists a set $W \subset V(H)$ such that $W$ induces an independent set in $H$.  The construction for this case is very similar, and in fact we can define our new graph $\constructind$ in terms of the first construction given above.

Note that, as $W$ induces an independent set in $G$, it must be that $W$ induces a clique in $\overline{H}$.  Thus we can apply the construction above to $G$ and $\overline{H}$ to obtain a graph $\constr(G,f_G,\overline{H},W)$.  We define $\constructind$ to be the complement of this graph, that is,
$$\constructind = \overline{\constr(G,f_G,\overline{H},W)}.$$
Once again, we equip our new graph with a colouring; in this case we set 
$$f_{\constructind} = f_{\constr(G,f_G,\overline{H},W)},$$
so the colouring is in fact exactly the same as that used in the case that $H$ contains a clique.

\subsection{Properties of the construction}
\label{const-props}

In this section we prove a number of important results about our constructions, which will be essential for the proofs in Section \ref{hard} below.  We begin by proving the key property of our constructions; we consider first the case for $\construct$.

\begin{lma}
Set $\widetilde{G} = \construct$, and let $X$ be a colourful subset of $\widetilde{G}$ with respect to $f_{\widetilde{G}}$.  Then the subgraph of $\widetilde{G}$ induced by $X$ is isomorphic either to $H$ or to a graph obtainable from $H$ by deleting one or more edges from $H[U]$.  Moreover, the number of edges deleted is equal to the number of non-edges in $\widetilde{G}[X \cap V_G]$.
\label{lose-U-edges}
\end{lma}
\begin{proof}
We begin by defining a bijection $\theta$ from $X$ to $V_H$; we will then argue that in fact $\theta$ defines an isomorphism from $\widetilde{G}[X]$ to a graph $H'$, where either $H' = H$, or else $H'$ can be obtained from $H$ by deleting one or more edges from $H[U]$.  The mapping $\theta$ is defined as follows:
$$\theta(x) = f_H^{-1}(f_{\widetilde{G}}(x)),$$
so each vertex $x \in X$ is mapped to the vertex of $H$ that receives the same colour under $f_H$.  Note that this is well-defined as $f_H$ is a bijection; the fact that $X$ is colourful implies that $f_{\widetilde{G}}|_{X}$ is also bijective and hence that $\theta$ is a bijection.

In order to show that there exists some graph $H'$ which satisfies the conditions of the lemma and is such that $\theta$ defines an isomorphism from $\widetilde{G}[X]$ to $H'$, it suffices to check that, for any two vertices $x,y \in X$ such that at least one of $\theta(x)$ and $\theta(y)$ does not lie in $U$, we have $xy \in E(\widetilde{G})$ if and only if $\theta(x)\theta(y) \in E(H)$.

Suppose first that both $\theta(x)$ and $\theta(y)$ lie in $V_H \setminus U$.  Then, by definition of the colouring $f_{\widetilde{G}}$, we must have $x,y \in V_H'$, and moreover $\theta(x) = x$ and $\theta(y) = y$; thus it follows immediately from the construction that $xy \in E(\widetilde{G})$ if and only if $\theta(x)\theta(y) \in E(H)$.

Now suppose that $\theta(x) \in U$, but $\theta(y) \notin U$.  Then, as before, we see that $\theta(y) = y$.  By definition of $\widetilde{G}$, the edge $xy$ belongs to $E(\widetilde{G})$ if and only if there is some vertex $w \in U$ such that $wy \in E(H)$ and $f_H(w) = f_G(x)$.  However, it follows from the definitions of $\theta$ and $f_{\widetilde{G}}$ that $f_G(x) = f_{\widetilde{G}}(x) = f_H(\theta(x))$, so $xy \in E(\widetilde{G})$ if and only if there is a vertex $w \in U$ such that $wy \in E(H)$ and $f_H(w) = f_H(\theta(x))$.  Since $f_H$ is injective, this is only possible if in fact $\theta(x) = w$, in other words $xy \in E(\widetilde{G})$ if and only if $\theta(x)\theta(y) \in E(H)$, as required.

Thus we see that there is indeed some suitable graph $H'$ such that $\theta$ defines an isomorphism from $\widetilde{G}[X]$ to $H'$.  The fact that $\theta$ is an isomorphism from $\widetilde{G}[X]$ to $H'$ implies that, for all $x,y \in X$ such that $\theta(x),\theta(y) \in U$, we have $\theta(x)\theta(y) \in E(H')$ if and only if $xy \in E(\widetilde{G})$.  Since the vertices that map to $U$ under $\theta$ are precisely those in $X \cap V_G$, this implies that the number of edges in $H'[U]$ is equal to the number of edges in $\widetilde{G}[X \cap V_G]$; hence (as $H[U]$ is complete) the number of edges we must delete from $H$ to obtain $H'$ is precisely equal to the number of non-edges in $\widetilde{G}[X \cap V_G]$, as required.
\end{proof}

It is now straightforward to derive the analogous result in the second case, for $\constructind$.

\begin{lma}
Set $\widehat{G} = \constructind$, and let $X$ be a colourful subset of $\widehat{G}$ with respect to $f_{\widehat{G}}$.  Then the subgraph of $\widehat{G}$ induced by $X$ is isomorphic either to $H$ or to a graph obtainable from $H$ by adding one or more edges to $H[W]$.  Moreover, the number of edges added is equal to the number of edges in $\widehat{G} \cap V_G]$.
\label{add-W-edges}
\end{lma}
\begin{proof}
Suppose first that $X$ is a colourful subset of $\widetilde{G}' = \constr(G,f_G,\overline{H},W)$ under $f_{\widetilde{G}'}$.  It follows from Lemma \ref{lose-U-edges} that the subgraph of $\widetilde{G}'$ induced by $X$ is isomorphic either to $\overline{H}$ or to a graph obtainable from $\overline{H}$ by deleting one or more edge with both endpoints in $W$.  Moreover, in this case the number of edges deleted is equal to the number of non-edges in $\widetilde{G}'[X \cap V_G]$.  The result follows immediately by taking complements.
\end{proof}

We now use this pair of results to prove some further facts about our constructions.  The first is an easy corollary.  Recall that, if $H$ is a graph and $\mathcal{H}$ a collection of labelled graphs, then $\mathcal{H}^H = \{(H',\pi'): H' \cong H\}$.
\begin{cor}
Let $k,k' \in \mathbb{N}$ with $k \leq k'$, and let $\mathcal{H}_{k'}$ be a collection of labelled graphs on $k'$ vertices, with $(H,\pi) \in \mathcal{H}_{k'}$.  Set $\widetilde{G} = \construct$ and $\widehat{G} = \constructind$.  Then,
\begin{enumerate}
\item if $U \in V(H)^{(k)}$ induces a clique in $H$ and there is no $(H',\pi') \in \mathcal{H}_{k'}$ such that $H'$ can be obtained from $H$ by deleting one or more edges in $U$, then
$$\ColStrEmb(\mathcal{H}_{k'}, \widetilde{G}, f_{\widetilde{G}}) = \ColStrEmb(\mathcal{H}_{k'}^H,\widetilde{G}, f_{\widetilde{G}}).$$
\item if $W \in V(H)^{(k)}$ induces an independent set in $H$ and there is no $(H',\pi') \in \mathcal{H}_{k'}$ such that $H'$ can be obtained from $H$ by adding one or more edges in $W$, then
$$\ColStrEmb(\mathcal{H}_{k'}, \widehat{G}, f_{\widehat{G}}) = \ColStrEmb(\mathcal{H}_{k'}^H,\widehat{G}, f_{\widehat{G}}).$$
\end{enumerate}
\label{clique-iso}
\end{cor}
\begin{proof}
We begin with the first part.  By definition, we know that the image of any mapping that contributes to $\ColStrEmb(\mathcal{H}_{k'},\widetilde{G},f_{\widetilde{G}})$ must be a colourful subset of $\widetilde{G}$ with respect to $f_{\widetilde{G}}$; but by Lemma \ref{lose-U-edges}, since no labelled graph in $\mathcal{H}_{k'}$ is isomorphic to a graph obtainable from $H$ by deleting one or more edges in $U$, any such subset must in fact be isomorphic to $H$.

The second part of the result follows by the same argument.
\end{proof}

The final fact we prove about our constructions is that the number of colourful subsets of $\construct$ (respectively $\constructind$) inducing copies of $H$ is equal to the number of colourful cliques in $G$.

\begin{lma}
Suppose that $U \in V(H)^{(k)}$ induces a clique.  Then, writing $\widetilde{G} = \construct$, 
$$\ColClique_k(G,f_G) = \ColSubInd(H,\widetilde{G},f_{\widetilde{G}}).$$
Similarly, if $W \in V(H)^{(k)}$ induces an independent set then, writing $\widehat{G} = \constructind$, 
$$\ColClique_k(G,f_G) = \ColSubInd(H,\widehat{G},f_{\widehat{G}}).$$
\label{count-cliques-stables}
\end{lma} 
\begin{proof}
We begin by showing that every colourful subset in $\widetilde{G}$ that induces a copy of $H$ corresponds to a distinct colourful clique in $G$.  Observe that, by Lemma \ref{lose-U-edges}, any colourful subset $X$ of $\widetilde{G}$ must induce a graph $H'$ that is either isomorphic to $H$ or else is obtainable from $H$ by deleting some edges in $H[U]$.  Moreover, the number of non-edges of $\widetilde{G}[X \cap V_G]$ is equal to the number of edges that must be deleted from $H$ to obtain $H'$.  Thus, if $X$ in fact induces a copy of $H$, then there cannot be any non-edges in $\widetilde{G}[X \cap V_G]$; in other words, $\widetilde{G}[X \cap V_G]$ is a clique.  By definition of $\widetilde{G}$, this means that $X \cap V_G$ induces a clique in $G$.  Note that, as $X$ is a colourful subset of $\widetilde{G}$ and colours from $[k]$ only appear at vertices from $V_G$ under $f_{\widetilde{G}}$, $\widetilde{G}[X \cap V_G]$ must in fact be a colourful clique with respect to the colouring $f_G$ (as $f_{\widetilde{G}}$ agrees with $f_G$ on $V_G$).  Now, observe that all colourful subsets $X$ must contain every vertex in $V_H'$, and so distinct colourful subsets $X$ and $X'$ must have distinct intersections with $V_G$.  Thus every colourful subset of $\widetilde{G}$ that induces a copy of $H$ corresponds to a distinct colourful clique in $G$.

Now we show that every colourful clique in $G$ corresponds to a distinct colourful subset in $\widetilde{G}$ that induces a copy of $H$.  Suppose that $Y$ induces a colourful clique in $G$ (with respect to the colouring $f_G$).  Observe that the set $Y \cup V_H'$ is colourful under $f_{\widetilde{G}}$, so by Lemma \ref{lose-U-edges} we know that $Y \cup V_H'$ induces a graph $H'$ that is either isomorphic to $H$ or to a subgraph of $H$ obtained by deleting one or more edges from $H[U]$.  Moreover, we know that the number of edges we must delete from $H$ to obtain $H'$ is equal to the number of non-edges in $\widetilde{G}[(Y \cup V_H') \cap V_G] = \widetilde{G}[Y]$.  Since $Y$ induces a clique in $G$, there are no non-edges in $\widetilde{G}[Y]$, and it must be that in fact $Y \cup V_H'$ induces a copy of $H$ in $\widetilde{G}$.  Finally, it is clear that distinct colourful cliques in $G$ will give distinct colourful copies of $H$.

The second part of the result, for $\constructind$, now follows easily by taking complements.
\end{proof}

\section{Hardness results}
\label{hard}

In this section we prove our results about the hardness of certain classes of parameterised subgraph counting problems.  We begin in Section \ref{aux} with some auxiliary results, then in Section \ref{bounded} we consider the case in which the property holds for a decreasing proportion of the possible edge densities, before giving a stronger result in Section \ref{intervals} for the special case in which the property depends only on the number of edges present in a subgraph.

\subsection{Auxiliary results}
\label{aux}

We prove two key lemmas which will be used throughout the rest of this section.  We begin by relating the number of subsets that induce a copy of a graph $H$ to the number of strong embeddings of graphs from a class of labelled graphs all isomorphic to $H$.

\begin{lma}
Let $\mathcal{H}$ be a collection of labelled graphs, and $(H,\pi) \in \mathcal{H}$ a labelled $k$-vertex graph.  Set 
\begin{align*}
\alpha_H = |\{\sigma : \quad & \sigma \text{ a permutation on $[k]$, $\exists (H,\pi') \in \mathcal{H}^H$ such that} \\
				    		 &  \text{$\pi' \circ \sigma^{-1} \circ \pi^{-1}$ defines an automorphism on $H$}\}|.
\end{align*}
Then, for any graph $G$, 
$$\StrEmb(\mathcal{H}^H,G) = \alpha_H \cdot \SubInd(H,G),$$
Moreover, if $G$ is equipped with a $k$-colouring $f$, then
$$\ColStrEmb(\mathcal{H}^H,G,f) = \alpha_H \cdot \ColSubInd(H,G,f).$$
\label{emb->subg}
\end{lma}
\begin{proof}
Observe first that $k$-tuples whose elements form the set $X \in V^{(k)}$ can only contribute to the quantity $\StrEmb(\mathcal{H}^H,G)$ if in fact $G[X] \cong H$.  We will argue that each subset that induces a copy of $H$ gives rise to exactly $\alpha_H$ tuples that contribute to $\StrEmb(\mathcal{H}^H,G)$.

Suppose that $X$ is such a subset; without loss of generality we may write $X = \{x_1,\ldots,x_k\}$ where the vertices are ordered so that
\begin{equation}
G[x_1,\ldots,x_k] = (H,\pi).
\label{pi-labelling}
\end{equation}

It is clear that there is a one-to-one correspondence between $k$-tuples whose elements form the set $X$ and permutations of $[k]$: we may regard the permutation $\sigma$ as corresponding to the tuple $(x_{\sigma(1)},\ldots,x_{\sigma(k)})$.  Observe that the tuple $(x_{\sigma(1)},\ldots,x_{\sigma(k)})$ will contribute to the value of $\StrEmb(\mathcal{H}^H,G)$ if and only if there exists some bijection $\pi':[k] \rightarrow V(H)$ such that $(H,\pi') \in \mathcal{H}^H$ and 
\begin{equation}
G[x_{\sigma(1)},\ldots,x_{\sigma(k)}] = (H,\pi').
\label{count-tuple}
\end{equation}
So the tuple contributes if and only if, for every $i,j \in [k]$, we have
$$x_{\sigma(i)}x_{\sigma(j)} \in E(G) \iff \pi'(i)\pi'(j) \in E(H).$$
Note that, by \eqref{pi-labelling}, for any $i,j \in [k]$ we have $x_{\sigma(i)}x_{\sigma(j)} \in E(G)$ if and only if $\pi(\sigma(i))\pi(\sigma(j)) \in E(H)$.  Thus \eqref{count-tuple} is equivalent to the statement that, for all $i,j \in [k]$,
$$ \pi(\sigma(i))\pi(\sigma(j)) \in E(H) \iff \pi'(i)\pi'(j) \in E(H),$$
which holds if and only if $\pi' \circ \sigma^{-1} \circ \pi^{-1}$ defines an automorphism on $H$.  Hence the number of $k$-tuples drawn from $X$ that contribute to the value of $\StrEmb(\mathcal{H}^H,G)$ is exactly $\alpha_H$.  Distinct subsets inducing $H$ will give rise to disjoint sets of $k$-tuples, so we see that in fact 
$$\StrEmb(\mathcal{H}^H,G) = \alpha_H \cdot \SubInd(H,G),$$
as required.

Exactly the same reasoning can be applied if we restrict to subsets $U$ that are colourful, which gives the second part of the result.
\end{proof}

Next we exploit the properties of our construction demonstrated in the previous section to give a sufficient condition for a parameterised subgraph-counting problem to be \#W[1]-complete.  We say that $(H,\pi) \in \mathcal{H}_{\phi_k}$ is \emph{good for $k'$-cliques} if there exists $U \in V(H)^{(k')}$ that induces a clique and there is no $(H',\pi') \in \mathcal{H}_{\phi_k}$ such that $H'$ can be obtained from $H$ by deleting edges with both endpoints in $U$.  Correspondingly, we say that $(H,\pi) \in \mathcal{H}_{\phi_k}$ is \emph{good for $k'$-independent sets} if there exists  $W \in V(H)^{(k')}$ that induces an independent set and there is no $(H',\pi') \in \mathcal{H}_{\phi_k}$ such that $H'$ can be obtained from $H$ by adding edges with both endpoints in $W$.

\begin{lma}
Let $\Phi$ be a family $(\phi_1,\phi_2,\ldots)$ of functions $\phi_k: \mathcal{L}(k) \rightarrow \{0,1\}$, such that the function mapping $k \mapsto \phi_k$ is computable.  Suppose there exists a computable function $g$ such that, for each $k' \in \mathbb{N}$, there exists $k \in \mathbb{N}$ with $k' \leq k \leq g(k')$ and $(H,\pi) \in \mathcal{H}_{\phi_k}$ that is either good for $k'$-cliques or is good for $k'$-independent sets.  Then \paramcount{\genprob}$(\Phi)$ is \#W[1]-complete under fpt Turing reductions.
\label{hardness-condition}
\end{lma}
\begin{proof}
We prove this result by means of an fpt Turing reduction from \paramcount{Multicolour Clique}.  Recall from Lemma \ref{uncol-col} that, for any $\Phi$, we have \paramcount{\genprobcol}($\Phi$) \leqfptT \paramcount{\genprob}($\Phi$), so it suffices to prove that \paramcount{MulticolourClique} \leqfptT \paramcount{\genprobcol}($\Phi$).

Suppose $(G,f_G)$ is an instance of \paramcount{MulticolourClique}, where $G$ is a graph and $f_G$ is a $k'$-colouring of the vertices of $G$.  By assumption, we can fix $k \in \mathbb{N}$ with $k' \leq k \leq g(k')$ such that some $(H,\pi) \in \mathcal{H}_{\phi_k}$ is good for either $k'$-cliques or $k'$-independent sets.  Recall also the definition of $\alpha_H$ from Lemma \ref{emb->subg}:
\begin{align*}
\alpha_H = |\{\sigma : \quad & \sigma \text{ a permutation on $[k]$, $\exists (H,\pi') \in \mathcal{H}^H$ such that} \\
				    		 &  \text{$\pi' \circ \sigma^{-1} \circ \pi^{-1}$ defines an automorphism on $H$}\}|.
\end{align*}
Now, if $(H,\pi)$ is good for $k'$-cliques (so some $U \in V(H)^{(k')}$ induces a clique, and there is no $(H',\pi') \in \mathcal{H}_{\phi_k}$ such that $H'$ can be obtained from $H$ by deleting edges with both endpoints in $U$), we observe that, setting $\widetilde{G} = \construct$,
\begin{align*}
\ColStrEmb(\mathcal{H}_{\phi_k},\widetilde{G},f_{\widetilde{G}})
											  & = \ColStrEmb(\mathcal{H}_{\phi_k}^H,\widetilde{G},f_{\widetilde{G}})
											   & \mbox{by Corollary \ref{clique-iso}} \\
											  & = \alpha_H \cdot \ColSubInd(H,\widetilde{G},f_{\widetilde{G}}) 
											   & \mbox{by Lemma \ref{emb->subg}} \\
											  & = \alpha_H \cdot \ColClique_{k'}(G,f_G) 
											   & \mbox{by Lemma \ref{count-cliques-stables}.}
\end{align*}
If instead $(H,\pi)$ is good for $k'$-independent sets, a symmetric argument implies that 
$$\ColStrEmb(\mathcal{H}_{\phi_k},\widehat{G},f_{\widehat{G}})  = \alpha_H \cdot \ColClique_{k'}(G,f_G),$$
where $\widehat{G} = \constructind$.
Thus, to compute the number of colourful cliques in $G$ under the colouring $f_G$, it suffices to perform the following steps.
\begin{enumerate}
\item Identify a suitable value of $k$ and a labelled graph $(H,\pi) \in \mathcal{H}_{\phi_k}$: this can be done by exhaustive search in time bounded only by some computable function of $k'$, as we know there exists a suitable $(H,\pi) \in \mathcal{H}_{\phi_k}$ for some $k \leq g(k')$.
\item Construct $\construct$, or $\constructind$, as appropriate: this can clearly be done in time polynomial in $n$ and $k \leq g(k')$.
\item Compute $\alpha_H$: this depends only on the graph $H$, so can be done in time bounded by some computable function of $k'$.
\item Perform a single oracle call to \paramcount{\genprobcol}$(\Phi)$: note that the parameter value is at most $g(k')$.
\end{enumerate}
This therefore gives an fpt Turing reduction from \paramcount{Multicolour Clique} to \paramcount{\genprobcol}$(\Phi)$, as required.
\end{proof}

\subsection{Properties that hold only for a decreasing proportion of the possible edge densities}
\label{bounded}

Suppose that we fix a family $\Phi = (\phi_1,\phi_2,\ldots)$ of functions $\phi_k: \mathcal{L}(k) \rightarrow \{0,1\}$, such that the function mapping $k \mapsto \phi_k$ is computable.  For each $k \in \mathbb{N}$, let $D_k = \{|E(H)|: (H,\pi) \in \mathcal{H}_{\phi_k}\}$, so $|D_k|$ is informally the number of distinct edge densities for which the property holds.  We will show that if $|D_k| = o(k^2)$ then \paramcount{\genprob}$(\Phi)$ is \#W[1]-complete.  

We need one auxiliary result before giving our main hardness result.

\begin{lma}
Let $\mathcal{H}_{k}$ be a non-empty collection of labelled graphs on $k$ vertices, where $k \geq 2^{2k'}$.  Suppose that 
$$r = |\{d: \exists (H,\pi) \in \mathcal{H}_k \text{ such that } |E(H)| = d\}|$$ 
satisfies 
\begin{equation}
r \leq \frac{1}{\binom{k-2}{k'-2}\binom{k'}{2}} \left( \frac{(2^{2k'}-k')!}{(2^{2k'})!} \frac{k!}{(k-k')!} \right).
\label{r-cond}
\end{equation}
Then there exists $(H,\pi) \in \mathcal{H}_k$ that is either good for $k'$-cliques or good for $k'$-independent sets. 
\label{bounded-layers}
\end{lma}
\begin{proof}
First recall from the corollary to Ramsey's Theorem (Corollary \ref{ramsey-cor}) that any graph on $k$ vertices must contain at least $\frac{(2^{2k'}-k')!}{(2^{2k'})!} \frac{k!}{(k-k')!}$ subsets of $k'$ vertices, where each of these subsets induces either a clique or an independent set.  It therefore follows immediately from \eqref{r-cond} that any graph on $k$ vertices must contain at least 
$$r \binom{k-2}{k'-2} \binom{k'}{2}$$
such $k'$-vertex subsets.

Now, for each $(H,\pi) \in \mathcal{H}_k$, let $\clique(H)$ denote the number of $k'$-cliques in $H$, and let 
$$\theta_{\mathcal{H}_k}(H,\pi) = \max_{\substack{(H',\pi') \in \mathcal{H}_k \\ |E(H')| = |E(H)|}} \clique(H').$$
We also set
$$C = \{\theta_{\mathcal{H}_k}(H,\pi) : (H,\pi) \in \mathcal{H}_k\}.$$

Observe that, if there is some $(H,\pi) \in \mathcal{H}_k$ such that $H$ contains at least one $k'$-clique $U$ and, for all $(H',\pi') \in \mathcal{H}_k$, we have $|E(H')| \geq |E(H)|$, then it is clear that $(H,\pi)$ would be good for $k'$-cliques (since, by edge-minimality, no labelled graph in $\mathcal{H}_k$ is isomorphic to any graph obtainable by deleting edges from $H$).  Thus we may assume from now on that every element $(H,\pi) \in \mathcal{H}_k$ with the minimum number of edges has $\clique(H) = 0$; it follows that for any such $(H,\pi)$ we in fact have $\theta_{\mathcal{H}_k}(H,\pi) = 0$.  Note that this implies we must have $0 \in C$.

Similarly, if there is some $(H,\pi) \in \mathcal{H}_k$ such that $H$ contains at least one independent set $W$ on $k'$ vertices and, for all $(H',\pi') \in \mathcal{H}_k$, we have $|E(H')| \leq |E(H)|$, then it is clear that $(H,\pi)$ must be good for $k$-independent sets (since, by edge-maximality, no labelled graph in $\mathcal{H}_k$ is isomorphic to any graph obtainable by adding edges to $H$).  Thus we may assume from now on that every edge-maximal element $(H,\pi) \in \mathcal{H}_k$ has no independent set on $k'$ vertices and so, by choice of $k$, satisfies $\clique(H) \geq r \binom{k-2}{k'-2} \binom{k'}{2}$.  Thus $C$ must contain an element $x$ where $x \geq r \binom{k-2}{k'-2} \binom{k'}{2}$.

Hence we may assume that $0 \in C$ and that the maximum element in $C$ is at least $r \binom{k-2}{k'-2} \binom{k'}{2}$; moreover, by definition of $r$ and $C$, we know that $C$ contains at most $r$ distinct values.  Thus, if the elements of $C$ are listed in order, there is some pair of consecutive elements which differ by more than $\binom{k-2}{k'-2} \binom{k'}{2}$; in other words, there exists some integer $s$ such that 
\begin{enumerate}
\item $s \in C$, and
\item for $s+1 \leq t \leq s + \binom{k-2}{k'-2} \binom{k'}{2}$, $t \notin C$.
\end{enumerate}
Fix $s \in \mathbb{N}$ satisfying these two conditions.  From now on we will say that a graph $(H,\pi) \in \mathcal{H}_k$ has ``few'' cliques if $\clique(H) \leq s$, and that it has ``many'' cliques if $\clique(H) > s + \binom{k-2}{k'-2} \binom{k'}{2}$.  By the reasoning above, it follows that, for every $(H,\pi) \in \mathcal{H}_k$, at least one of the following must hold:
\begin{enumerate}
\item $(H,\pi)$ has few cliques, or
\item $\theta_{\mathcal{H}_k}(H,\pi) > s + \binom{k-2}{k'-2} \binom{k'}{2}$, so there is some $(H',\pi') \in \mathcal{H}_k$ such that $|E(H)| = |E(H')|$ and $(H',\pi')$ has many cliques.
\end{enumerate}

Now, we fix an element $(H,\pi)$ with as few edges as possible from those graphs in $\mathcal{H}_k$ that contain many cliques (so $(H,\pi)$ contains many cliques and, for any other $(H',\pi')$ that contains many cliques, $|E(H)| \leq |E(H')|$).  This choice of $(H,\pi)$ implies that any element of $\mathcal{H}_k$ with strictly fewer edges than $H$ must contain few cliques.

Fix a set $U \in V(H)^{(k)}$ that induces a clique in $H$.  Suppose that some element $(H',\pi') \in \mathcal{H}_k$ is such that $H'$ can be obtained from $H$ by deleting one or more edges with both endpoints in $U$.  Since we will then have $|E(H')| < |E(H)|$, it follows that $H'$ contains few cliques.  Hence there are at least $\binom{k-2}{k'-2} \binom{k'}{2} + 1$ more $k'$-cliques in $H$ than in $H'$; as these two graphs differ only in edges that have both endpoints in $U$, it must be that each $k'$-clique in $H$ that is not a $k'$-clique in $H'$ intersects $U$ in at least two vertices.  But the number of $k'$-vertex sets that intersect $U$ in at least two vertices is at most $\binom{k'}{2}\binom{k-2}{k'-2}$, so 
it is not possible for $\binom{k-2}{k'-2} \binom{k'}{2} +1$ distinct $k'$-cliques in $H$ each to intersect $U$ in at least two vertices, giving a contradiction.

Thus we see that $(H,\pi)$ must in fact be good for $k'$-cliques, completing the proof.
\end{proof}

We are now ready to prove \#W[1]-hardness for this class of problems.

\begin{thm}
Let $\Phi$ be a family $(\phi_1,\phi_2,\ldots)$ of functions $\phi_k: \mathcal{L}(k) \rightarrow \{0,1\}$, infinitely many of which are not identically zero, such that the function mapping $k \mapsto \phi_k$ is computable.  Suppose that $|D_k| = o(k^2)$.  Then \paramcount{\genprob}($\Phi$) is \#W[1]-complete under fpt Turing reductions.
\label{bounded-hard}
\end{thm}
\begin{proof}
We exploit Lemma \ref{hardness-condition} to prove the result.  By the assumption that $|D_k| = o(k^2)$, we know that, for any fixed $\alpha$, there exists $k_0 \in \mathbb{N}$ such that, for any $k \geq k_0$, we have $|D_k| < \alpha k(k-1)$.  Setting $\alpha =  \frac{(2^{2k'} - k')!(k'-2)!}{\binom{k'}{2} (2^{2k'})!} $, we therefore see that, for sufficiently large $k$,
\begin{align*}
|D_k| & \leq \left( \frac{(2^{2k'} - k')!(k'-2)!}{\binom{k'}{2} (2^{2k'})!} \right) k(k-1) \\
      & = \frac{1}{\binom{k-2}{k'-2}\binom{k'}{2}} \left( \frac{(2^{2k'}-k')!}{(2^{2k'})!} \frac{k!}{(k-k')!} \right).
\end{align*}
Set $g(k')$ to be the least such $k$ (and note that, by computability of the mapping $k \mapsto \phi_k$, $g$ is computable: we can perform an exhaustive search to find a suitable $k$).  Note that, for this value of $k$, $\mathcal{H}_{\phi_k}$ satisfies the condition of Lemma \ref{bounded-layers}.  Hence we know that there exists $(H,\pi) \in \mathcal{H}_{\phi_k}$ that is either good for $k'$-cliques or good for $k'$-independent sets.
\#W[1]-hardness now follows immediately from Lemma \ref{hardness-condition}.
\end{proof}

\subsection{Properties that are defined by $o(k^2)$ intervals of permitted edge densities}
\label{intervals}

In this section we give a strengthening of the above result for the special case in which the property $\Phi$ depends only on the number of edges in the subgraph.  In the following theorem, we will be considering \emph{integer intervals}, that is, sets of consecutive integers (e.g.~$\{a,a+1,\ldots,b\}$).

\begin{thm}
Let $\Phi$ be a family $(\phi_1,\phi_2,\ldots)$ of functions $\phi_k: \mathcal{L}(k) \rightarrow \{0,1\}$, infinitely many of which are not identically zero, such that the function mapping $k \mapsto \phi_k$ is computable.  For each $k$, let $\mathcal{I}_k = \{I_1,\ldots,I_r\}$ be a collection of disjoint integer intervals, where each $I_i \subset \{0,\ldots,\binom{k}{2}\}$, $\emptyset \neq \bigcup \mathcal{I}_k \neq \{0,\ldots,\binom{k}{2}\}$, and $|\mathcal{I}_k| = o(k^2)$.  Suppose that $\phi_k(H,\pi) = 1$ if and only if $|E(H)| \in \bigcup \mathcal{I}_k$.  Then \paramcount{\genprob}($\Phi$) is \#W[1]-complete.
\label{interval-hard}
\end{thm}
\begin{proof}
We claim that we may assume, without loss of generality, that $|\bigcup \mathcal{I}_k| \leq \frac{1}{2} (\binom{k}{2} + 1)$.  To see that this is indeed the case, suppose that in fact $|\bigcup \mathcal{I}_k| > \frac{1}{2} (\binom{k}{2} + 1)$, and consider the family $\Phi' = (\phi_1',\phi_2',\ldots)$ of functions $\phi_k': \mathcal{L}(k) \rightarrow \{0,1\}$ defined by 
$$\phi_k' = 1 - \phi_k.$$
Note that the mapping $k \mapsto \phi_k'$ is clearly computable by computability of $k \mapsto \phi_k$, and that there exists a collection $\mathcal{I}_k'$ of disjoint integer intervals, where each $I_i' \subset \{0, \ldots, \binom{k}{2}\}$, $\emptyset \neq \bigcup \mathcal{I}_k' \neq \{0,\ldots,\binom{k}{2}\}$, $|\mathcal{I}_k'| = o(k^2)$ and $\bigcup \mathcal{I}_k' = \{0,\ldots,\binom{k}{2}\} \setminus \bigcup \mathcal{I}_k$.  Thus $\phi_k'(H,\pi) = 1$ if and only if $|E(H)| \in \bigcup \mathcal{I}_k'$, and $|\bigcup \mathcal{I}_k'| \leq \frac{1}{2} (\binom{k}{2} + 1)$.  By Proposition \ref{complement}, it therefore suffices to prove \#W[1]-completeness in the case that $|\bigcup \mathcal{I}_k| \leq \frac{1}{2} (\binom{k}{2} + 1)$.

We do this using Lemma \ref{hardness-condition}.  Since $|\mathcal{I}_k| = o(k^2)$, it follows that, for any $k' \in \mathbb{N}$, there exists $k \in \mathbb{N}$ such that $(|\mathcal{I}_k| + 1)\binom{k'}{2} + 1 < \frac{1}{2} (\binom{k}{2} + 1)$; we define $g(k')$ to be the least such $k$ (note that under this definition the function $g$ is clearly computable).  In order to apply Lemma \ref{hardness-condition} to show \#W[1]-hardness, it suffices to demonstrate that there exists $(H,\pi) \in \mathcal{H}_{\phi_k}$ which satisfies one of the two conditions in the statement of Lemma \ref{hardness-condition}.

Note that $\{0,\ldots,\binom{k}{2}\} \setminus \bigcup \mathcal{I}_k$ must be expressible as the union of at most $|\mathcal{I}_k| + 1$ disjoint integer intervals; hence, as 
$$|\{0,1,\ldots,\tbinom{k}{2}\} \setminus \bigcup \mathcal{I}_k | \geq \tfrac{1}{2} (\tbinom{k}{2} + 1) \geq (|\mathcal{I}_k| + 1)\tbinom{k'}{2} + 1,$$
it follows that at least one of these integer intervals, $J$, must contain at least $\binom{k'}{2} + 1$ distinct integers.  

Suppose first that $0 \notin J$.  Then there exists some $d_1 \in \{0,1,\ldots,\binom{k}{2}\}$ such that $d_1 \in \bigcup \mathcal{I}_k$ but $d_1 + 1 \in J$.  Note that, as $J$ contains at least $\binom{k'}{2} + 1$ distinct integers, we must have $d_1 < \binom{k}{2} - \binom{k'}{2}$.  Thus there exists a labelled graph $(H,\pi) \in \mathcal{L}(k)$ with $d_1$ edges that contains an independent set $W$ on $k'$ vertices; since by assumption all elements of $\mathcal{L}(k)$ with exactly $d_1$ edges belong to $\mathcal{H}_{\phi_k}$, we therefore have $(H,\pi) \in \mathcal{H}_{\phi_k}$.  However, as there is no $(H',\pi') \in \mathcal{H}_{\phi_k}$ with $|E(H)| < |E(H')| \leq |E(H)| + \binom{k'}{2}$, it is clear that there is no $(H',\pi') \in \mathcal{H}_{\phi_k}$ which can be obtained from $H$ by adding edges in $W$.  Thus we satisfy the second condition of Lemma \ref{hardness-condition}.

Now suppose that $0 \in J$.  Since $\bigcup \mathcal{I}_k \neq \emptyset$, we must have $\binom{k}{2} \notin J$, and so there must exist some $d_2 \in \{0,1,\ldots,\binom{k}{2}\}$ such that $d_2 \in \bigcup \mathcal{I}_k$ but $d_2 - 1 \in J$.  Note that, as $J$ contains at least $\binom{k'}{2} + 1$ distinct integers, we must have $d_2 > \binom{k'}{2}$.  Thus there exists a labelled graph $(H,\pi) \in \mathcal{L}(k)$ with $d_2$ edges that contains a clique $U$ on $k'$ vertices; since by assumption all elements of $\mathcal{L}(k)$ with exactly $d_2$ edges belong to $\mathcal{H}_{\phi_k}$, we therefore have $(H,\pi) \in \mathcal{H}_{\phi_k}$.  However, as there is no $(H',\pi') \in \mathcal{H}_{\phi_k}$ with $|E(H)| - \binom{k'}{2} \leq |E(H')| < |E(H)|$, it is clear that there is no $(H',\pi') \in \mathcal{H}_{\phi_k}$ which can be obtained from $H$ by deleting edges in $U$.  Thus we satisfy the first condition of Lemma \ref{hardness-condition}.

Hence we see that there must be some $(H,\pi) \in \mathcal{H}_{\phi_k}$ which satisfies at least one of the conditions of Lemma \ref{hardness-condition}; this immediately implies the \#W[1]-hardness of \paramcount{\genprob}$(\Phi)$ in this case.
\end{proof}

\section{Conclusions and Open Problems}

We have proved \#W[1]-completeness for a range of parameterised subgraph-counting problems.  In particular, we demonstrated that \paramcount{\genprob}($\Phi$) is \#W[1]-complete whenever $\Phi$ is such that one of the following holds:
\begin{itemize}
\item $|\{|E(H)|: (H,\pi) \in \mathcal{H}_{\phi_k}\}| = o(k^2)$, or
\item $\phi_k(H,\pi) = 1$ if and only if $|E(H)| \in \bigcup \mathcal{I}_k$, where $\mathcal{I}_k$ is a collection of integer intervals contained in $\{0,\ldots,\binom{k}{2}\}$ and $|\mathcal{I}_k| = o(k^2)$.
\end{itemize}
These results extend some existing hardness results concerning parameterised subgraph-counting problems, and additionally include, for example, the problems of counting planar subgraphs, subgraphs with treewidth at most $t$ for any fixed $t$, and regular subgraphs, as well as the problem of counting $k$-vertex subgraphs with at least $d(k)$ edges, for any function $d$ where $0<d(k)<\binom{k}{2}$.

A natural question arising from the second class of problems we consider is whether all non-trivial properties that depend only on the number of edges in the subgraph are in fact \#W[1]-hard, or whether there might exist a fixed parameter algorithm for some such problems that are not covered by our result, such as counting the number of $k$-vertex subgraphs having an even number of edges.

It should be noted that the methods used to demonstrate hardness in this paper are based on the hardness of the multicolour version of the problem (demonstrated for appropriate $\Phi$ in Lemma \ref{hardness-condition}) and so are only applicable to problems \paramcount{\genprob}($\Phi$) where  \paramcount{\genprobcol}($\Phi$) is also \#W[1]-hard.  However, there are known examples of \#W[1]-complete parameterised counting problems whose multicolour versions are in fact fixed parameter tractable, such as \paramcount{Path}, \paramcount{Cycle} and \paramcount{Matching} (the multicolour versions of these problems are all fixed parameter tractable by \cite{arvind02}, as they involve counting embeddings of graphs of bounded treewidth).  A challenge for future research, therefore, would be to develop new kinds of constructions that can be used to show hardness of problems whose multicolour versions are fixed parameter tractable.

\bibliographystyle{amsplain}
\bibliography{../param_counting_refs}

\providecommand{\bysame}{\leavevmode\hbox to3em{\hrulefill}\thinspace}
\providecommand{\MR}{\relax\ifhmode\unskip\space\fi MR }
\providecommand{\MRhref}[2]{%
  \href{http://www.ams.org/mathscinet-getitem?mr=#1}{#2}
}
\providecommand{\href}[2]{#2}
\begin{thebibliography}{10}

\bibitem{arvind02}
V.~Arvind and Venkatesh Raman, \emph{Approximation algorithms for some
  parameterized counting problems}, ISAAC 2002 (P.~Bose and P.~Morin, eds.),
  LNCS, vol. 2518, Springer-Verlag Berlin Heidelberg, 2002, pp.~453--464.

\bibitem{chen07}
Yijia Chen and J.~Flum, \emph{On parameterized path and chordless path
  problems}, Computational Complexity, 2007. CCC '07. Twenty-Second Annual IEEE
  Conference on, 2007, pp.~250--263.

\bibitem{chen08}
Yijia Chen, Marc Thurley, and Mark Weyer, \emph{Understanding the complexity of
  induced subgraph isomorphisms}, Automata, Languages and Programming (Luca
  Aceto, Ivan Damg{\aa}rd, Leslie~Ann Goldberg, Magn{\'u}s~M. Halld{\'o}rsson,
  Anna Ing{\'o}lfsd{\'o}ttir, and Igor Walukiewicz, eds.), Lecture Notes in
  Computer Science, vol. 5125, Springer Berlin Heidelberg, 2008, pp.~587--596.

\bibitem{radu13}
Radu Curticapean, \emph{Counting matchings of size k is \#{W}[1]-hard},
  Automata, Languages, and Programming (Fedor~V. Fomin, R{\=u}si{\c n}{\v s}
  Freivalds, Marta Kwiatkowska, and David Peleg, eds.), Lecture Notes in
  Computer Science, vol. 7965, Springer Berlin Heidelberg, 2013, pp.~352--363.

\bibitem{radu14}
Radu Curticapean and D{\'a}niel Marx, \emph{Complexity of counting subgraphs:
  Only the boundedness of the vertex-cover number counts}, 55th Annual IEEE
  Symposium on Foundations of Computer Science, FOCS 2014, 2014.

\bibitem{dalmau04}
V\'{i}ctor Dalmau and Peter Jonsson, \emph{The complexity of counting
  homomorphisms seen from the other side}, Theoretical Computer Science
  \textbf{329} (2004), no.~1–3, 315 -- 323.

\bibitem{erdos-szekeres}
P.~Erd{\H{o}}s and G.~Szekeres, \emph{A combinatorial problem in geometry},
  Compositio Math. \textbf{2} (1935), 464--470.

\bibitem{fellows09}
M.~Fellows, D.~Hermelin, F.~Rosamond, and S.~Vialette, \emph{On the
  parameterized complexity of multiple-interval graph problems}, Theoretical
  Computer Science \textbf{410} (2009), 53--61.

\bibitem{flum04}
J.~Flum and M.~Grohe, \emph{The parameterized complexity of counting problems},
  SIAM Journal on Computing \textbf{33} (2004), no.~4, 892--922.

\bibitem{flumgrohe}
\bysame, \emph{Parameterized complexity theory}, Springer, 2006.

\bibitem{connected}
Mark Jerrum and Kitty Meeks, \emph{The parameterised complexity of counting
  connected subgraphs and graph motifs}, arXiv:1308.1575v2[cs.CC], August 2013.

\end{thebibliography}

% History dates
%\received{}{}{}

\end{document}